\documentclass[reprint,twocolumn,showpacs,prd]{revtex4-1}
\usepackage{amsmath,amssymb,amsfonts,amsthm}
\newtheorem{theorem}{Theorem}[section]
\newtheorem{lemma}[theorem]{Lemma}
\newtheorem{definition}[theorem]{Definition}

\newtheorem{conjecture}[theorem]{Conjecture}
\usepackage{graphicx}

\newcommand{\vv}{\bar \sigma}

\newcommand{\vY}{\bar \omega}

\newcommand{\mf}{\mathcal{M}}

\newcommand{\mq}{m_{bh}}

\newcommand{\dv}{dS}
\newcommand{\ds}{ds}

\newcommand{\Ls}{\Delta_0}

\newcommand{\vi}{\iota}

\begin{document}
\title{Extreme throat initial data set and horizon area--angular
  momentum inequality for axisymmetric black holes}

\author{Sergio Dain} \affiliation{Facultad de Matem\'atica, Astronom\'{\i}a y
  F\'{i}sica, FaMAF, Universidad
  Nacional de C\'ordoba,\\
  Instituto de F\'{\i}sica Enrique Gaviola, IFEG, CONICET,\\
  Ciudad Universitaria (5000) C\'ordoba, Argentina} \affiliation{Max Planck
  Institute for Gravitational Physics (Albert Einstein Institute) Am
  M\"uhlenberg 1 D-14476 Potsdam Germany.}

\date{\today}

\begin{abstract}
  We present a formula that relates the variations of the area of extreme
  throat initial data with the variation of an appropriate defined mass
  functional. From this expression we deduce that the first variation, with
  fixed angular momentum, of the area is zero and the second variation is
  positive definite evaluated at the extreme Kerr throat initial data. This
  indicates that the area of the extreme Kerr throat initial data is a minimum
  among this class of data.  And hence the area of generic throat initial data
  is bounded from below by the angular momentum.  Also, this result strongly
  suggests that the inequality between area and angular momentum holds for
  generic asymptotically flat axially symmetric black holes. As an application,
  we prove this inequality in the non trivial family of spinning Bowen-York
  initial data.
\end{abstract}

\maketitle

\section{Introduction}
\label{sec:introduction}

Geometrical inequalities play an important role in General Relativity, in
particular for vacuum black holes, where the geometrical aspects of the theory
appears in their pure form. A geometrical inequality in General Relativity
relates quantities that have both a physical interpretation and a geometrical
definition.  The most relevant example is the positive mass theorem. The mass
of the spacetime measures the total amount of energy and hence it should be
positive from the physical point of view. Also, the mass $m$ in General
Relativity is represented by a pure geometrical quantity on a Riemannian
manifold \cite{Arnowitt62} \cite{Bartnik86}.  From the geometrical mass
definition, without the physical picture, it would be very hard to conjecture
that this quantity should be positive.  On the other hand, the highly non
trivial proof of this inequality \cite{Schoen79b}\cite{Schoen81}\cite{Witten81}
reveals the subtle way in which Einstein equations describe the gravitational
field.

For black holes, the first example of geometrical inequality is the Penrose
inequality (see the recent review article \cite{Mars:2009cj} and references
therein) which relates the
area of the horizon $A$ (the `size' of the black holes) with the total mass of
the spacetime
\begin{equation}
  \label{eq:65}
 \sqrt{\frac{A}{16\pi}}\leq m.
\end{equation}
Another example is the inequality between mass
and the angular momentum $J$ for axially symmetric black holes 
\begin{equation}
  \label{eq:4}
  \sqrt{|J|}\leq m. 
\end{equation}
See \cite{Dain06c}\cite{Dain05e} \cite{Chrusciel:2007ak}, and also the
generalization which includes charge presented in \cite{Chrusciel:2009ki}
\cite{Costa:2009hn}.  Inequalities \eqref{eq:65} and \eqref{eq:4} are closed
related with the weak cosmic censorship conjecture. They can be interpreted as
indirect but relevant indications of the validity of this conjecture.  As in
the case of the positive of the mass, these inequalities has been discovered
first by physical arguments, and then proved afterword (under appropriated and
restricted assumptions, see the references mentioned above) as a rigorous
consequence of Einstein equations.  The proofs provide also new insight into
the mechanisms of Einstein equations. As an example (which is connected with
the subject of this article), we mention that the proof of \eqref{eq:4}
involves a variational characterization of the extreme Kerr black hole as an
absolute minimum of the mass.


The total mass is a global quantity. On the other hand the area $A$ and the
angular momentum in axial symmetry $J$, involved in inequalities \eqref{eq:65}
and \eqref{eq:4} respectively, are quasi-local quantities. Namely they carry
information on a bounded region of the spacetime. In contrast with a local
quantity like a tensor field which depends on a point of the spacetime.
Inequalities \eqref{eq:65} and \eqref{eq:4} relate global quantities with
quasi-local ones. It is well known that the energy of the gravitational field
can not be represented by a local quantity. The best one can hope is to obtain
an expression of the total energy of a bounded region of the spacetime. These
are the so called quasi-local mass definition (see the review article
\cite{Szabados04} and reference therein). For some of the quasi-local mass
proposals there exist also  positivity proofs and hence we obtain a
quasi-local geometrical inequality.  One would expect that for a black hole
pure quasi-local inequalities are also valid.  The relevance of this kind of
inequalities is that they provide a much finer control on the dynamics of a black
holes than the global versions.  The main purpose of this article is to study
one example of such quasi-local inequality for vacuum black holes and to give
non trivial evidences of its validity.


The area of the horizon is a well defined quasi-local quantity for a generic
black hole. In general, the quasi-local angular momentum is difficult to define
(see \cite{Szabados04}), but in the case of axial symmetry there exists a well
defined notion, namely the Komar angular momentum.  This is the angular
momentum used in inequality \eqref{eq:4}.  That is, for generic (i.e. not necessarily
stationary) axially symmetric black holes we have two well defined quasi-local
physical quantities, the horizon area $A$ and the angular momentum $J$. In
terms of $A$ and $J$, the Christodoulou \cite{Christodoulou70} mass of the
black hole is defined as follows
\begin{equation}
  \label{eq:32}
  \mq=\sqrt{\frac{A}{16\pi}+ \frac{4\pi J^2}{A}}.
\end{equation}
For the Kerr black hole this formula gives precisely the total mass of the
black hole which is equal to the total mass of the spacetime. In general, for
an horizon with only one connected component one would expect that $\mq$ is
less than the total mass of the spacetime with equality only for the Kerr black
hole.  This would give a generalization of Penrose inequality including angular
momentum (in fact, an strong generalization) in the spirit of \cite{Jang79}
(see the discussion in \cite{Mars:2009cj}).  For the case of many black holes
the negative interaction energy between the black holes should be also
considered and hence it is not expected a simple inequality with respect to the
total mass (see \cite{Weinstein05} for a discussion of the analog of this
inequality with charges). We will not discuss this point further, since we
are interested in this article only in quasi-local inequalities. We mention it
because it plays a role in the physical interpretation of the quasi-local mass
$\mq$.

The formula \eqref{eq:32} trivially satisfies the inequality
\begin{equation}
  \label{eq:33}
   \sqrt{|J|}\leq \mq. 
\end{equation}
This is, of course, just because the Kerr black hole satisfies this bound.
Hence, if we accept \eqref{eq:32} as the correct formula for the quasi-local
mass of an axially symmetric black hole, then \eqref{eq:33} provides the,
rather trivial, quasi-local version of \eqref{eq:4}. The real question is, does
the formula \eqref{eq:32} represents the quasi-local mass of a non-stationary
black hole? Let us analyze the behavior of the mass \eqref{eq:32} from a physical point
of view. 

Let us assume that for a generic axially symmetric black hole the quantity
$\mq$ gives a measure of the quasi-local mass of the black hole.  Consider the
evolution of $\mq$. By the area theorem, we know that the horizon area will
increase. If we assume axial symmetry, then the angular momentum will be
conserved at the quasi-local level (we are assuming pure vacuum). On physical
grounds, one would expect that in this situation the quasi-local mass of the
black hole should increase with the area, since there is no mechanism at the
classical level to extract mass from the black hole.  In effect, the only way
to extract mass from a black hole is by extracting angular momentum through a
Penrose process.  But angular momentum transfer is forbidden in pure vacuum
axial symmetry.  Then, one would expect that both the area $A$ and the
quasi-local mass $\mq$ should monotonically increase with time.

Let us take a time derivative of $\mq$ (denoted by a dot). From the formula
\eqref{eq:32} we obtain
\begin{equation}
  \label{eq:34}
  \dot m_{bh} = \frac{\dot A}{32\pi \mq} \left(1-\left(\frac{8\pi
        J}{A}\right)^2 \right), 
\end{equation}
were we have used that the angular momentum $J$ is conserved. Since, by the
area theorem, we have
\begin{equation}
  \label{eq:9}
  \dot A \geq 0,
\end{equation}
the time derivative of $\mq$ will be positive (and hence the mass $\mq$ will
increase with the area) if and only if the following inequality is satisfied
\begin{equation}
  \label{eq:5}
8\pi |J|\leq  A.
\end{equation}
Then, it is natural to conjecture that \eqref{eq:5} should be satisfied for any
horizon in an axially symmetric asymptotically flat initial data.  If there are
initial data that violate \eqref{eq:5} then in the evolution the area will
increase but the mass $\mq$ will decrease. This will indicate that the quantity
$\mq$ has not the desired physical meaning.  Also, a rigidity statement is
expected. Namely, the equality in \eqref{eq:5} is reached only by the extreme
Kerr black hole.

If inequality \eqref{eq:5} is true, then we have a non trivial monotonic
quantity (in addition to the black hole area) $\mq$
\begin{equation}
  \label{eq:10}
 \dot m_{bh} \geq 0.
\end{equation}

It is important to emphasize that the physical arguments presented above in
support of \eqref{eq:5} are certainly weaker in comparison with the ones behind
the inequalities \eqref{eq:65} and \eqref{eq:4}. A counter example of any of
these inequality will prove that the standard picture of the gravitational
collapse is wrong. On the other hand, a counter example of \eqref{eq:5} will
just prove that the quasi-local mass \eqref{eq:32} is not appropriate to
describe the evolution of a non-stationary black hole.  One can 
imagine other expressions for quasi-local mass, may be more involved, in axial
symmetry.  On the contrary, reversing the argument, a proof of \eqref{eq:5}
will certainly prove that the mass \eqref{eq:32} has physical meaning for
non-stationary black holes as a natural quasi-local mass. Also, the inequality
\eqref{eq:5} provide a non trivial control of the size of a black hole valid at
any time.  

If the rigidity statement also holds, this inequality will provide a remarkable
quasi-local measure of how far is the data from the extreme black hole data.
This provides an `extremality criteria' in the spirit of \cite{Booth:2007wu},
although restricted only to axial symmetry.  In the article \cite{Dain:2007pk}
it has been conjectured that, within axially symmetry, to prove the stability
of a nearly extreme black hole is perhaps simpler than a Schwarzschild black
hole. It is possible that this quasi-local extremality criteria will have relevant
applications in this context.

Let us also point out that the inequality \eqref{eq:5} is related with the
surface gravity density (or temperature) of a black hole.  The surface gravity
density $\kappa$ of the Kerr black hole can be written  in terms of the
quasi-local quantities $A$ and $J$ as follows
\begin{equation}
  \label{eq:14}
  \kappa=\frac{1}{4\sqrt{\frac{A}{16\pi}+\frac{4\pi J^2}{A}}}
  \left(1-\left(\frac{8\pi J}{A}\right)^2 \right). 
\end{equation}
From this equation, we see that for the Kerr black hole $\kappa$ is positive
because inequality (\ref{eq:5}) holds. Hence, if inequality (\ref{eq:5}) holds
for  generic, non-stationary axially symmetric black holes we can define the
same expression for $\kappa$ for this class of black holes.

All the previous arguments lead to the following conjecture
\begin{conjecture}
\label{c:1}
Consider an asymptotically flat, vacuum, complete axially symmetric initial
data set for the Einstein equations. Then the following inequality holds
\begin{equation}
  \label{eq:16}
 8\pi |J|\leq  A,
\end{equation}
where $A$ and $J$ are the area and angular momentum of a connected component of
the apparent horizon.
\end{conjecture}
Note that in the previous discussion  we have considered the area $A$ of the
event horizon (since we have used the area theorem). As it usual in geometrical
inequalities in order to make an useful statement we need to replace the event
horizon with a quasi-local quantity. In our case the most appropriate quantity appears
to be the apparent horizon on the initial data. Generalization of the area
theorem also holds (under appropriate assumptions) for apparent horizons (see
the review article \cite{lrr-2004-10} and reference therein).

Let us mention some support for this inequality. This inequality has been
proved for stationary black holes surrounded by matter in \cite{hennig08}
\cite{Hennig:2008zy} (also with charge). Although this case is only slightly
related with the conjecture (since the conjecture applies for non-stationary
spacetimes in vacuum) it is highly non trivial and it certainly suggests the
validity the conjecture.  It is also important to note that there is a counter
example for this inequality in the non-asymptotically flat case
\cite{Booth:2007wu}. That is, the assumption that the horizon belong to an
asymptotically flat data is essential. This counter example points out  that
although the inequality involves only quasi-local quantities is not pure
quasi-local in the sense that a global assumption should be made on the initial
data (namely, asymptotic flatness).

The purpose of this article is to present non trivial evidences for the
validity of conjecture \ref{c:1}. The main part of this evidence is a formula
that relates in a remarkable way the variations of the area and the variations
of an appropriate defined mass functional on extreme throat  initial data.
These kind of data (described in section \ref{sec:extr-cylindr-init}) isolate
the cylindrical structure of extreme black hole and hence they represent a
natural source of counter examples to the inequality \eqref{eq:16} as we will
see.  The very existence of this formula suggests that the inequality
\eqref{eq:16} should hold, at least in a relevant family of initial
conditions. Using this result we also prove the inequality in the non trivial
family of spinning Bowen-York initial data.

The plan of the article is the following. In section
\ref{sec:extr-cylindr-init} we describe extreme throat initial data. In section
\ref{sec:main-result} we present our main result, given by theorem
\ref{t:main}. The proof of this result is divided naturally in two steps,
described en section \ref{sec:mass-funct-extr} and \ref{sec:vari-area-extr}. In
section \ref{sec:mass-funct-extr} we present an appropriate mass functional for
extreme throat initial data. We calculate the first and second variations of this
functional evaluated at the extreme Kerr cylindrical initial data. These results
are analogous to the ones described for asymptotically flat axially symmetric
initial data described in \cite{Dain05c} \cite{Dain05d}. Section
\ref{sec:vari-area-extr} constitutes the most important part of the article. In
this section it is shown the relation between the variations of the mass
functional and the variations of the area. This relation was, in our opinion,
completely unexpected a priori. In section \ref{sec:aplic-spinn-bowen} we apply
this result to prove the inequality \eqref{eq:16} on the spinning Bowen-York
black hole family. We discuss the relevant open problems in section
\ref{sec:final-comments}. Finally, we conclude with an appendix in which we collect
some properties of the extreme Kerr throat initial data.
 
\section{Extreme throat  initial data set}
\label{sec:extr-cylindr-init}

In order to present our results we need to discuss first extreme throat 
initial data. The definition of this kind of initial data is motivated by the
behavior of the Kerr black hole initial data in the extreme limit. Let us
briefly review this behavior (for more details, see for example, section 2
in \cite{Dain:2010uh}).

Consider the Kerr black hole with mass $m$ and angular momentum $J$.  We define
the following parameter $\mu$ (which has unit of mass) in terms of $m$ and $J$
\begin{equation}
  \label{eq:mu}
  \mu=\sqrt{ m^2 -|J|}.
\end{equation} 
The extreme Kerr black hole corresponds to $\mu = 0$. For the Schwarzschild
black hole we have $\mu=m$.

In the standard Boyer-Lindquist coordinates for the Kerr black hole, take a
slice $t=constant$.  Let us denote by $S$ the 3-dimensional manifold defined by
that slice.  The topology of this surface is $S=S^2\times \mathbb{R}$.
The triple $(S,h_{ij}, K_{ij})$, where $ h_{ij}$ is the induced intrinsic
metric on $S$ and $ K_{ij}$ is the second fundamental form of $S$, constitute
an initial data set for Einstein equations.  That is, they are solutions of the
constraint equations
\begin{align}
 \label{const1}
   D_j   K^{ij} -  D^i   K= 0,\\
 \label{const2}
   R -  K_{ij}   K^{ij}+  K^2=0,
\end{align}
where $ {D}$ and $ R$ are the Levi-Civita connection and the Ricci scalar
associated with ${h}_{ij}$, and $ K = K_{ij} h^{ij}$. In these equations the
indices are moved with the metric $ h_{ij}$ and its inverse $ h^{ij}$.

For $\mu >0$ these data have the geometry of two asymptotically flat ends. In
the extreme limit $\mu =0$ the geometry changes. One of the asymptotic ends is
asymptotically flat but the other is cylindrical. 
Let us take a closer look at the structure of the cylindrical end.
In coordinates ($r,\theta,\phi$), the induced metric on $S$ has the form
\begin{equation}\label{thcero}
 h_{ij}=\Phi^4\tilde h_{ij},
\end{equation}
where the conformal metric $\tilde h_{ij}$ is defined by
\begin{equation}
  \label{eq:42}
  \tilde h=e^{2q}(dr^2+r^2d\theta^2)+r^2\sin^2\theta d\phi^2,
\end{equation}
and the functions $\Phi$ and $q$ are given by equations \eqref{ficero}  in  appendix
\ref{sec:extr-kerr-cylindr}. The extrinsic curvature is given by 
\begin{equation}
\label{tkcero}
 K_{ij}=\frac{2}{\eta} S_{(i} \eta_{j)}, \quad
 S_i=\frac{1}{\eta}\epsilon_{ijk}\eta^j\partial^k\omega,
\end{equation}  
where $\eta^i$ is the axial Killing vector 
\begin{equation}
  \label{eq:7}
\eta^i=\frac{\partial}{\partial \phi},  
\end{equation}
 the square of its norm $\eta$ is  given by  \eqref{eq:59bb}, $\epsilon_{ijk}$
denotes the volume element with respect to the metric $h_{ij}$ and $\omega$ is
given by \eqref{wcero}. The advantage of this particular form of writing
$K_{ij}$ is that it is easy to check from \eqref{tkcero} that $K_{ij}$
satisfies the momentum constraint \eqref{const1} (see, for example, the
appendix in \cite{Dain99b} and \cite{Dain:2010uh}). In particular, we have
that $K_{ij}$ is trace-free, namely
\begin{equation}
  \label{eq:62}
   K=0.
\end{equation}
That is, these initial data are maximal surfaces. 

In these coordinates, the asymptotically flat end of the metric \eqref{thcero}
corresponds to the limit $r\to \infty$ and the cylindrical end corresponds to
the limit $r\to 0$.  The radial coordinate $r$ is a good coordinate in the
asymptotically flat end, since the metric and the extrinsic curvature are
manifestly asymptotically flat with respect to these coordinates: they have
the standard decay to the flat metric.

On the other hand, in the limit $r \to 0$ the conformal factor $\Phi$ blows
up. This is, however, just a coordinate problem. To see this, define $s=-\ln r$,
then the cylindrical end corresponds to $s\to \infty$, and the metric has the
form
\begin{equation}
  \label{eq:2}
    h^0=(\sqrt{r}\Phi)^4\left(e^{2q}(ds^2+d\theta^2)+\sin^2\theta
     d\phi^2 \right).
\end{equation}
The functions $\sqrt{r}\Phi$ and $q$ are smooth and uniformly bounded in
the whole range $-\infty < s< \infty$.

The metric
\eqref{eq:2} and the second fundamental form \eqref{tkcero} have a well defined
limit $s\to \infty$ as initial data. For the metric $h^0$ in the limit $s\to
\infty$  we obtain 
\begin{equation}
  \label{eq:55}
  h=\varphi_0^4(e^{2q_0}(ds^2+d\theta^2)+\sin^2\theta d\phi^2).
\end{equation}
where $\varphi_0$ and $q_0$ are defined by the limits
\eqref{eq:56}--\eqref{eq:57}. The extrinsic curvature $K^{ij}$ has the form
\eqref{tkcero} where $\omega$ is replaced by its limiting value $\omega_0$
defined by \eqref{eq:58} and all the other quantities are computed with respect
to the metric \eqref{eq:55}.  These are in fact solutions of the constraint
equations \eqref{const1}--\eqref{const2} on $S^2\times \mathbb{R}$. 
We call these initial data set the extreme Kerr throat initial data set.

Let us make a summary.  The extreme Kerr throat  initial data set is constructed
out of the Kerr black hole initial data by two limits. 
The first one is the extreme limit
\begin{equation}
  \label{eq:3}
  \mu \to 0.
\end{equation}
In this limit the geometry of the Kerr black hole initial data changes from two
asymptotically flat ends to one asymptotically flat and one cylindrical. The
second limit is
\begin{equation}
  \label{eq:15}
  s\to \infty. 
\end{equation} 
This limit isolates the cylindrical structure of the extreme Kerr initial data
cutting off the asymptotically flat end.

This procedure of taking the extreme limit can be perform for more generic data
(see \cite{Dain:2008yu}, \cite{gabach09}). And the behavior is identical,
although, of course, one ends up with a different extreme throat initial data.  

We isolate the properties of generic extreme throat initial data in the following
definition. Consider the following Riemannian metric
\begin{equation}
  \label{eq:1}
  h=\varphi^4(e^{2q}(ds^2+d\theta^2)+\sin^2\theta d\phi^2),
\end{equation}
were the functions $\varphi$ and $q$ depend only on $\theta$. 
We assume that $\varphi$ and $q$ satisfy the following equation
\begin{equation}
  \label{eq:2b}
  \Ls\varphi -\frac{1}{4}(1-\partial^2_\theta
  q)\varphi=-\frac{|\partial_\theta \omega|^2}{16\sin^4\theta \varphi^7},
\end{equation}
where $\Ls$ is the Laplace operator in $S^2$ with respect to the standard
metric acting on axially symmetric functions, namely
\begin{equation}
  \label{eq:17}
   \Ls\varphi = \frac{1}{\sin\theta} \partial_\theta
   \left(\sin\theta \partial_\theta \varphi  \right).
\end{equation}

Finally, for convenience, we define out of $\varphi$ two additional functions,
$\sigma$ and $\eta$, as follows
\begin{equation}
  \label{eq:40}
  \varphi^4=e^\sigma,\quad  \eta=\sin^2\theta \varphi^4. 
\end{equation}
The function $\eta$ is the square of the norm of the Killing vector (\ref{eq:7})
 with respect to the metric \eqref{eq:1}.
 
With these ingredients, we can formulate the following definition. 
\begin{definition}
  Consider a set of functions $(\sigma,\omega, q)$ (depending only on $\theta$)
  that satisfy equation \eqref{eq:2b} on $S^2$ and such that $q$ vanished at
  the poles $\theta=0,\pi$. Then, an extreme throat initial data set is a
  triple $(S, h_{ij}, K_{ij})$ where $S=\mathbb{R}\times S^2$, $h_{ij}$ is
  given by \eqref{eq:1} and $K_{ij}$ is constructed from $\omega$ by the
  formula \eqref{tkcero}. In equation \eqref{tkcero}, the volume element
  $\epsilon_{ijk}$ is calculated with respect to the metric \eqref{eq:1}, the
  vector $\eta^i$ is given by (\ref{eq:7}) and the indices are moved  with
  the metric \eqref{eq:1}.
\end{definition}

From the definition it follows that the data satisfy the constraint equations
(\ref{const1})--(\ref{const2}), since equation \eqref{eq:2b} is just the
Lichnerowicz equation for the conformal factor $\varphi$ with respect to the
conformal metric
\begin{equation}
  \label{eq:19}
  \tilde h =e^{2q}(ds^2+d\theta^2)+\sin^2\theta d\phi^2).
\end{equation}
Note that the Ricci scalar of $\tilde h$ is given by
\begin{equation}
  \label{eq:68}
  \tilde R= 2e^{-2q}(1-\partial^2_\theta q).
\end{equation}
For a discussion on the Lichnerowicz equation and the conformal method see, for
example, the review article \cite{Bartnik04b}.

The vector $\eta^i$ is a Killing vector of the metric $h_{ij}$ 
\begin{equation}
  \label{eq:12}
   \pounds_\eta h_{ij}=0,
\end{equation}
where $\pounds$ denotes Lie derivative.   Moreover, the
requirement that $q$ vanishes at the poles arises from the regularity of the
metric at the axis, namely the condition
\begin{equation}
  \label{eq:8}
  \lim_{\theta\to 0,\pi} \frac{\partial_i\eta \partial^i \eta}{4\eta}=1.
\end{equation}
Hence, the metric $h_{ij}$ is axially symmetric (see \cite{Dain06c}
\cite{Chrusciel:2007dd} for a discussion of axial symmetry on initial data and
also \cite{stephani03} for a general discussion of axial symmetry in General
Relativity). 

From the definition of $K_{ij}$ it is also clear that $\eta^i$ is a symmetry of
$K_{ij}$, namely
\begin{equation}
  \label{eq:11}
  \pounds_\eta K_{ij}=0.
\end{equation}
In our definition, we have made for simplicity the assumption that $\eta^i$ is
hypersurface orthogonal (with respect to the metric $h_{ij}$). We expect that
all the results obtained in this article are also valid without this
assumption, but this analysis remains to be done.

The metric $h_{ij}$ has another symmetry, namely the vector $\xi^i$ defined by
\begin{equation}
  \label{eq:13}
  \xi^i=\frac{\partial }{\partial s}.
\end{equation}
It straightforward to check that $\xi^i$ is also a symmetry of $K_{ij}$
\begin{equation}
  \label{eq:66}
   \pounds_\xi K_{ij}=0.
\end{equation}
Also, the vectors $\eta^i$ and $\xi^i$ commute. 

Riemannian metrics of the form \eqref{eq:1} are generically called
cylindrically symmetric. In addition, we have equations (\ref{eq:11}) and
(\ref{eq:66}) and hence at first sight it looks appropriate to call the whole
initial data set cylindrically symmetric. However this terminology is
misleading for the following reason: in general, the spacetime originated from
this kind of data will not be cylindrically symmetric. Recall that a spacetime
is cylindrically symmetric if it admits two spacelike commuting Killing vectors
(see \cite{stephani03}). Since the problem of cylindrically symmetric
spacetimes has been frequently analyzed in the literature, it is important to
discuss this point in detail.

The vectors $\eta^i$ and $\xi^i$ are Killing vectors of $h_{ij}$ and we also
have equations (\ref{eq:11})--(\ref{eq:66}), hence it follows from the results
of \cite{Moncrief75} that the development of this class of initial data will be
a spacetime with, at least, two Killing vectors.  The projection of the
spacetime Killing vectors to the initial surface are given by $\eta^i$ and
$\xi^i$ (see \cite{beig97} for a discussion about the relation of spacetime
symmetries and symmetries on the initial data).

These data constitute initial data for the axially symmetric vacuum Einstein
equations (see, for example, \cite{Dain:2008xr} \cite{dain10}), hence it
follows that the spacetime will be axially symmetric.  In particular, the
spacetime Killing vector $\eta^\mu$ corresponding to $\eta^i$ will be spacelike
outside the axis.

However, although the spacetime will have another symmetry it will not be
cylindrically symmetric, because the extra symmetry will not be, in general,
spacelike.  The behavior of the spacetime Killing vector $\xi^\mu$ originated
from the initial data symmetry $\xi^i$ is clearly illustrated in the following
explicit example which is also interesting by itself.

Consider the following 4-dimensional metric in coordinates $(t,s,\theta,\phi)$
\begin{multline}
  \label{eq:74}
  g= \frac{(1+\cos^2\theta)}{2}
  \left[-\frac{e^{-2s}}{r^2_0} dt^2  +r^2_0(ds^2+d\theta^2) \right]+\\
\eta_0 \left(d\phi
    +  \frac{e^{-s}}{r^2_0}  dt \right)^2,  
\end{multline}
where $r^2_0=2|J|$ and $\eta_0$ is given by (\ref{eq:6}). 
This metric was introduced in \cite{Bardeen:1999px} as the 
extreme Kerr throat geometry. It characterizes the spacetime geometry near the
horizon of the extreme Kerr black hole. 

It can be easily verified that the extreme Kerr throat initial data given by
equations (\ref{eq:55}) and (\ref{tkcero}) are the initial data of the metric
(\ref{eq:74}) in a surface $t=constant$. The spacetime Killing vectors of the
metric $g$ which correspond to the initial data Killing vectors $\eta^i$ and
$\xi^i$ are given by
\begin{equation}
  \label{eq:75}
  \eta^\mu=\frac{\partial }{\partial \phi},   \quad \xi^\mu =t\frac{\partial
  }{\partial t}-  \frac{\partial }{\partial s}. 
\end{equation}
The metric has also two more Killing vectors (see \cite{Bardeen:1999px})
\begin{equation}
  \label{eq:76}
  \xi_1^\mu = \frac{\partial }{\partial t},\quad   \xi_2^\mu=
  \left(\frac{e^{-2s}}{2} +\frac{t^2}{2} \right)    \frac{\partial
  }{\partial t}- t\frac{\partial }{\partial s}-
e^{-s}\frac{\partial }{\partial \phi}.
\end{equation}
We see that the Killing vector $\xi^\mu$ is not spacelike everywhere. In
particular, the metric $g$ is not cylindrically symmetric.

Finally, let us mention that there are two important physical quantities
defined on a extreme throat initial data. First, the angular momentum given by
\begin{equation}
  \label{eq:41}
  J=\frac{1}{8} \left(\omega(\pi)-\omega(0)\right).
\end{equation}
This formula follows from the expression of the angular momentum for standard
asymptotically flat 
axially symmetric initial data (see, for example, \cite{Dain06c}). 
Second, the area of the cylinder
\begin{equation}
  \label{eq:26}
  A =2\pi \int_0^\pi e^{\sigma+q} \sin\theta \, d\theta.
\end{equation}

\section{Main result}
\label{sec:main-result}
The extreme limit procedure \eqref{eq:3} and \eqref{eq:15} that lead to the
extreme throat initial data for the Kerr black hole discussed in the previous
section \ref{sec:extr-cylindr-init} has an additional, remarkable property.
The area of the extreme cylinder (with value $A=8\pi|J|$) is smaller that the
minimal surface area of any non-extreme Kerr black hole initial data (recall
that the angular momentum $J$ is kept fixed). In fact, the area of the minimal
surface is a monotonically decreasing function with respect to $\mu$. This can
be, of course, trivially verified since for the Kerr black hole we have the
explicit expression for $A$ in terms of $\mu$.

As we have pointed out, this extreme limit can be performed for other class of
initial data, like the Bowen-York black hole initial data showed in section
\ref{sec:aplic-spinn-bowen}. It is conceivable (but it certainly remain to be
shown) that there exists such procedure for general black hole initial data
in axial symmetry, or at least for a relevant family of initial data. Let us
assume that this is the case. That is, let as assume that for an initial data
with an horizon of area $A_1$ we can perform the limit procedure to obtain an
extreme  throat  initial data of area $A$, with $A\leq A_1$. Then, if
inequality \eqref{eq:16} is true, it should also holds for the extreme
throat initial data. Our main result indicates that this is precisely the case.
This result is summarized in the following theorem.

\begin{theorem}
\label{t:main}
  Let us consider families of extreme throat initial data with fixed angular
  momentum $J$. Then, the area on these families satisfy the
  following properties:
\begin{itemize}
\item The first variation of the area is zero evaluated on the extreme Kerr
throat initial data.

\item The second variation of the are is positive evaluated  on  the extreme Kerr
  throat initial data. 
\end{itemize}
  
\end{theorem}  

This theorem strongly suggests that the area is an absolute minimum for extreme
Kerr throat initial data among all the extreme throat initial data with the
same angular momentum. Since for extreme Kerr we have $A=8\pi |J|$, the
inequality \eqref{eq:16} is satisfied for general extreme throat initial data.
In order to prove that, we can follow a similar line as in \cite{Dain05d} to
prove that it is a local minimum and to \cite{Dain06c} \cite{Costa:2009hn}
\cite{Chrusciel:2009ki} \cite{Chrusciel:2007ak} to prove that it is in fact a
global minimum.  It appears that the same analysis will go throw without major
difficulties.  This however should be checked and it will be done in a
subsequent work.

Theorem \ref{t:main} gives also strong evidences in favor to inequality
\eqref{eq:16}. Namely, if this inequality were false, there is no reason to
expect that it will hold on extreme throat initial data. As it have been pointed out
above, this theorem suggest also an strategy to prove the conjecture: given an
initial data with an apparent horizon construct a limit procedure analogous to
\eqref{eq:3} and \eqref{eq:15} in such a way that i) in the limit an extreme
throat initial data set is obtained and ii) the area of the extreme throat
initial data is less or equal than the area of the horizon.  In fact, in
section \ref{sec:aplic-spinn-bowen} we construct this limit procedure for the
spinning Bowen-York family of initial data.

The proof of theorem \ref{t:main} is naturally divided in two parts, presented
in sections \ref{sec:mass-funct-extr} and \ref{sec:vari-area-extr}
respectively.

\section{The mass functional for extreme throat initial data}
\label{sec:mass-funct-extr}
An extreme throat  initial data are stationary if the following equations are
satisfied
\begin{align}
  \label{eq:18}
  \Delta_0\sigma-2 & =-\frac{|\partial_\theta \omega|^2}{\eta^2}\\
  \label{eq:20}
  \Delta_0\omega  & =2\frac{\partial_\theta\omega\partial_\theta\eta}{\eta}.
\end{align} 
The fact that these equations for an extreme throat initial data define
stationary solutions can be deduced from the standard stationary axially
symmetric equations.  However, for our present purpose, the only property of
equations \eqref{eq:18}--\eqref{eq:20} that we  will use is that the extreme
 Kerr throat initial data (defined by \eqref{eq:35b}--\eqref{eq:58}) are a 
solution of them. This can be easily checked explicitly.

Equation \eqref{eq:20} can be written in divergence form as follows
\begin{equation}
  \label{eq:21}
  \partial_\theta \left(\sin\theta \frac{\partial_\theta \omega}{\eta^2}
  \right)=0.
\end{equation}
The stationary equations can be written in a natural form
as equations on the unit sphere $S^2$ with the standard metric. Namely, let
$D_A$ be the covariant derivative with respect to the standard metric in
$S^2$. Then, equations \eqref{eq:18}--\eqref{eq:20} are written as
\begin{align}
  \label{eq:43}
  D_A D^A\sigma -2 & =\frac{D_A\omega D^A\omega}{\eta^2},\\
D_A\left(\frac{D^A\omega}{\eta^2}\right) & =0.  \label{eq:43f}
\end{align}
These expression were defined for axially symmetric functions, but they
also make sense for functions which depends on $\phi$. In fact, in all the
results that follows we will not use the assumption that the functions are
axially symmetric (this is very similar to what happens in the study of
the inequality \eqref{eq:4} discussed in \cite{Dain06c}).

We define the following functional
\begin{equation}
  \label{eq:22}
 \mf= \int_0^\pi\left( |\partial_\theta \sigma|^2 +4\sigma + \frac{|\partial_\theta
    \omega|^2}{\eta^2}\right) \sin\theta \, d\theta.
\end{equation}
On the unit sphere, using the notation of equations
\eqref{eq:43}--\eqref{eq:43f} this functional is written as
\begin{equation}
  \label{eq:44}
  \mf=\frac{1}{2\pi} \int_{S^2}\left( |D \sigma|^2 +4\sigma + \frac{|D
      \omega|^2}{\eta^2}\right)  \, \dv,  
\end{equation}
were $\dv=\sin\theta\, d\theta d\phi$ is the volume element of the standard
metric in $S^2$.  This functional is the obvious translation of the mass
functional used in \cite{Dain05c} adapted to this kind of  initial
data. 

Let us make some general comments regarding the functional $\mf$ which are not
directly relevant for the present article but they can have interesting future
applications. It is very likely that for non-stationary initial data the
functional $\mf$ represents a lower bound for another mass functional $\mf'$
which includes the time dependent terms. This is what happens with the
functional considered in \cite{Dain05c}. When the complete spacetime is
considered (and not just the initial data), this new functional is precisely
the total energy (the ADM mass) of axially symmetric, asymptotically flat
spacetimes, and it is conserved (see \cite{Dain:2008xr}). In the present case,
the mass functional $\mf'$ will describe the total energy of the class of
spacetime discussed in section \ref{sec:extr-cylindr-init}.  Namely, axially
symmetric spacetimes which has another Killing vector.  These spacetimes are
not asymptotically flat.  An analog situation occur for cylindrical symmetric
spacetimes, for which the total energy can be defined (see
\cite{Ashtekar:1996cd} and reference therein). We emphasize however that the
situation here is more complicated since the extra Killing vector is not
spacelike everywhere. It would be very interesting to explore this issue and
construct explicitly the functional $\mf'$.

Relevant for our present purpose, are the following two important properties of
the mass functional \eqref{eq:22}. We will prove them in in lemma \ref{l:1}.
The first one is that the stationary equations are the Euler-Lagrange equations
of this functional.  That is, the extreme Kerr throat initial data are critical
points of this functional.  The second property is that the second variation of
this functional evaluated at the extreme Kerr throat initial data is positive.
That suggests that extreme Kerr throat initial data set is in fact a minimum of
this functional.  These properties can be expected from the analysis developed
in \cite{Dain05c} and \cite{Dain05d}, since the functional $\mf$ is the natural
generalization of the mass functional used in these articles adapted to
cylindrical initial data.

Before proving that lemma is important to make the connexion between the
functional $\mf$ and the energy of harmonic maps between $S^2$ and
$H^2$. Namely, consider the functional
\begin{equation}
  \label{eq:47}
\tilde  \mf_\Omega=  \frac{1}{2\pi} \int_\Omega \frac{|\partial \eta|^2+|\partial
    \omega|^2}{\eta^2}   \, \dv,
\end{equation}
defined on some domain $\Omega\subset S^2$, such that $\Omega$ does not include
the poles. Integrating by parts and using  the
identity
\begin{equation}
  \label{eq:48}
\Ls(\log(\sin\theta))=-1,  
\end{equation}
we obtain the following relation between $\mf$ and $\mf'$
\begin{multline}
  \label{eq:49}
  \tilde\mf_\Omega= \mf_\Omega+4 \int_\Omega \log\sin\theta\, dS+ \\
\oint_{\partial
    \Omega} (4\sigma + \log\sin\theta) \frac{\partial \log\sin\theta}{\partial
    n}\, \ds,
\end{multline}
where $n$ denotes the exterior normal to $\Omega$, $\ds$ is the surface element
on the boundary $\partial \Omega$ and we have used the obvious notation
$\mf_\Omega$ to denote the mass functional \eqref{eq:44} defined over the
domain $\Omega$. The difference between $\mf$ and $\mf'$ are the boundary
integral plus the second term which is just a numerical constant. Note that if
we integrate over $S^2$ this constant term  is finite
\begin{equation}
  \label{eq:50}
   \int_\Omega \log\sin\theta\, dS=2\log2-2.
\end{equation}
The boundary terms however diverges at the poles. 

In an analogous way as it was described in \cite{Dain06c}, the functional
$\mf'$ defines an energy for maps $(\eta,\omega):S^2\to \mathbb{H}^2$ where
$\mathbb{H}^2$ denotes the hyperbolic plane $\{(\eta, \omega ) : \eta > 0\}$,
equipped with the negative constant curvature metric
\begin{equation}
  \label{eq:51}
  ds^2=\frac{d\eta^2+d\omega^2}{\eta^2}. 
\end{equation}
The Euler-Lagrange equations for the energy $\mf'$ are called harmonic maps
from $S^2\to \mathbb{H}^2$. Since $\mf$ and $\mf'$ differ only by a constant
and boundary terms, they have the same Euler-Lagrange equations.

We present in the following lemma the main result of this section. 
\begin{lemma}
\label{l:1}
  Let us consider families of extreme throat initial  data with fixed angular
  momentum $J$. Then, the area on these families satisfy the
  following properties:
\begin{itemize}
\item The first variation of $\mf$ is zero evaluated on  the extreme Kerr
  throat initial 
  data.

\item The second variation of $\mf$ is positive evaluated on the extreme Kerr
  throat initial data.
\end{itemize}
  
\end{lemma}
\begin{proof}
  The proof follows very similar lines as the one presented in \cite{Dain05d}. The only
  difference is the presence of an extra term in $\mf$, the one containing
  $\sigma$. But this term, since it is linear, makes no contribution to the
  second variation which is the delicate part of the proof.

To define the variations, let us consider  the real-valued function
$\vi(\epsilon)$  
defined by 
\begin{equation}
  \label{eq:19b}
\vi(\epsilon)= \mf(\sigma(\epsilon),\omega(\epsilon)),
\end{equation}
where 
\begin{equation}
  \label{eq:35}
\sigma(\epsilon)=\sigma_0+\epsilon\vv,\quad \omega(\epsilon)=\omega_0+
\epsilon\vY. 
\end{equation}
We assume that $\bar \omega$ vanished at the poles
$\theta=0,\pi$. This boundary condition keeps fixed the angular momentum under
the variations. In analogous way we define
\begin{equation}
  \label{eq:33b}
 \eta(\epsilon)=\sin^2\theta e^{\sigma(\epsilon)}.  
\end{equation}

The  first derivative of $\vi(\epsilon)$ with respect to $\epsilon$ is
given by 
\begin{multline}
  \label{eq:43b}
\vi'(\epsilon) = \frac{1}{\pi}\int_{S^2}
  \left\{D_A \sigma D^A \vv + 2\vv+ \right. \\
\left. +\left ( D_A \omega  D^A \vY -\vv  |D \omega|^2
\right)\eta^{-2}\right\} \dv, 
\end{multline}
where a prime denote derivative with respect to $\epsilon$ and the $\epsilon$
dependence in the right-hand side of \eqref{eq:43b} is encoded in the functions
$\sigma(\epsilon),\omega(\epsilon),\eta(\epsilon)$ defined by
\eqref{eq:35}--\eqref{eq:33b}. If we evaluate at $\epsilon=0$, integrate by
parts and use the condition that $\bar \omega$ vanished at the poles we obtain
the Euler-Lagrange equations \eqref{eq:43}--\eqref{eq:43f}. Since extreme Kerr
is a solution of this equation the first item in the Lemma is proved.

The second derivative of $\vi$ is given by
\begin{multline}
  \label{eq:44b}
    \vi''(\epsilon)= \frac{1}{16\pi}\int_{S^2}
  \left\{| D \vv |^2 + \right.\\
\left. +\left(2\vv^2|D \omega|^2
 - 4\vv D_A \omega D^A \vY + |D \vY|^2
\right)\eta^{-2}\right\} \dv.  
\end{multline}
From equation \eqref{eq:44b}, it is far from obvious that the second variation
evaluated at the critical point $\epsilon=0$ is positive definite.  In order to
prove that, the key ingredient is the following remarkable identity proved by
Carter \cite{Carter71}. In terms of our variables it has the following form
\begin{equation}
  \label{eq:9b}
F + \vv G'_{\sigma}+\vY G'_{\omega}+2\vv\vY G_{\omega} - \eta^{-2}\vY^2 G_{\sigma} =H
\end{equation}
 where
\begin{align}
  \label{eq:200}
G_\sigma(\epsilon) & = \Delta \sigma + \eta^{-2}  |D \omega|^2 -2,\\
G_\omega(\epsilon) &= D_A( \eta^{-2}D^A \omega), \label{eq:200b}
\end{align}
the derivatives with respect to $\epsilon$ of these functions are given
\begin{align}
  \label{eq:101}
  G'_\sigma(\epsilon) & =\Delta \vv + \left( 2D_A \vY D^A \omega -2\vv  |D
    \omega|^2 
  \right)\eta^{-2},\\  
  G'_\omega(\epsilon) &= D_A \left(\eta^{-2} \left(D^A \vY-2\vv D^A \omega\right)
  \right), \label{eq:102}
\end{align}
the positive definite function $F$ is given by 
\begin{multline}
  \label{eq:26b}
F(\epsilon)=\left(D\vv+ \vY \eta^{-2} D \omega \right )^2 +
\left(D( \vY  \eta^{-1})- \eta^{-1} \vv D \omega \right)^2\\
+\left( \eta^{-1} \vv  D \omega  - \vY \eta^{-2} D \eta \right)^2,
\end{multline}
and the divergence term $H$ is given by
\begin{equation}
  \label{eq:52}
H= D_A \left(\vv D^A \vv +\vY \eta^{-1}   D^A \left(  \vY \eta^{-1}
  \right)\right),     
\end{equation}
The identity \eqref{eq:9b} is valid for arbitrary functions
$\sigma,\omega,\vv,\vY$ and it is straightforward to check although the
computations are lengthy.

Note that using \eqref{eq:101}--~\eqref{eq:102} and integrating by parts we
obtain
\begin{equation}
  \label{eq:11b}
-\int_{S^2} \left( \vv G'_{\sigma}(\epsilon)+\vY G'_{\omega}(\epsilon) \right)
\dv =\pi \vi''(\epsilon). 
\end{equation}  

We integrate on $S^2$ the identity \eqref{eq:9b}. The divergence term $H$
vanished (here we use again the boundary condition). We use \eqref{eq:11b} to
obtain
\begin{equation}
  \label{eq:53}
  \vi''(\epsilon)= \int_{S^2}F \, \dv+ \int_{S^2} \left( 2\vv\vY G_{\omega}(\epsilon) -
  \eta^{-2}\vY^2  G_{\sigma}(\epsilon)\right)\, \dv. 
\end{equation}
If we evaluate at $\epsilon=0$ the last integral vanished, and hence we get the
final result
\begin{equation}
  \label{eq:54}
   \vi''(0)= \int_{S^2}F \, \dv \geq 0.
\end{equation}

\end{proof}

The mass functional  $\mf$ evaluated at extreme Kerr gives the value 
 \eqref{eq:39} which in particular is not equal to the total mass $m$ of
 extreme Kerr. This is to be expected since there
 is no obvious relation between $\mf$ and the total mass of the associated  
   initial data with  an asymptotically flat end and a cylindrical
  end. However, the value of $\mf$ at extreme Kerr suggests the following definition
\begin{equation}
  \label{eq:45}
  m=Ce^\frac{{\mf}}{16}, \quad C= e^{-\frac{\ln(2)}{2}-\frac{1}{2}}. 
\end{equation}
We have normalized this quantity in such a way that gives the mass for extreme
Kerr. It is also trivially positive definite (note that $\mf$ is not due to the
extra term $\sigma$ which has no sign). More important, the first variation of
$m$ and the second variation of $m$ are given by
\begin{equation}
  \label{eq:46}
   m' =2^{-4} \mf' m, \quad m'' = 2^{-8}  ( \mf'' +( \mf')^2)m.
\end{equation}
And hence the functional $m$ has the same critical points as $\mf$ and the
second variation is also definitive positive.  These properties makes the
functional $m$ attractive but we will not make use of it in the following. For
the purpose of the proof of theorem \ref{t:main} only the functional $\mf$ is
used.

\section{Variation of the area for extreme throat initial  data}
\label{sec:vari-area-extr}

The results from previous section are somehow to be expected, since they are
the analogous of the variational formulation presented in \cite{Dain05c} and
\cite{Dain05d}.  The remarkable new ingredient is the relation of this mass
functional with the area. This is the subject of this section and it
constitutes the most relevant part of this article.

Consider the formula for the area for an extreme throat  initial data given
by \eqref{eq:26}. The first and second variation of the area are given 
\begin{equation}
  \label{eq:27}
   A' =  \int_{S^2}(\sigma'+ q')  e^{(\sigma+q)} \, dS,
\end{equation}
and 
\begin{equation}
  \label{eq:28}
  A'' =  \int_{S^2}((\sigma'+ q')^2 +( \sigma'' +  q''))
  e^{(\sigma+q)} \, dS. 
\end{equation}
In order to relate these equations with the mass functional we proceed as
follows. We first write the Hamiltonian constraint \eqref{eq:2b} in terms of
$\sigma$ using the relation \eqref{eq:40}
\begin{equation}
  \label{eq:69}
  4\Ls\sigma +|\partial_\theta \sigma|^2 + \frac{|\partial_\theta
    \omega|^2}{\eta^2}  -4(1-\partial^2_\theta q)=0. 
\end{equation}
We integrate this equation in $S^2$. The first term gives zero. We write the second
and third term in terms of  mass functional \eqref{eq:44}, namely
\begin{equation}
  \label{eq:70}
  \int_{S^2}|\partial_\theta \sigma|^2 + \frac{|\partial_\theta
    \omega|^2}{\eta^2} \, dS= 2\pi \mf- 4\int_{S^2}\sigma \, dS.
\end{equation}
For the last term, we integrate by part the terms with $\partial^2_\theta q$, namely
\begin{align}
  \label{eq:71}
  \int_{S^2} \partial^2_\theta q \, dS &= 2\pi \int_{0}^\pi \partial^2_\theta q
  \sin\theta \, d\theta\\
&= 2\pi \int_{0}^\pi\left(\partial_\theta (\partial_\theta q
  \sin\theta)- \partial_\theta q \cos \theta\right) \, d\theta  \label{eq:71a} \\
&= -2\pi \int_{0}^\pi \partial_\theta q \cos \theta \, d\theta \label{eq:71b}\\
&= 2\pi \int_{0}^\pi \left(-\partial_\theta(q\cos\theta) +q\sin\theta\right) \,
d\theta \label{eq:71c}\\
&= 2\pi \int_{0}^\pi q\sin\theta \, d\theta  \label{eq:71d}  \\
& = \int_{S^2} q \, dS.  \label{eq:71e}
\end{align}
To pass from \eqref{eq:71a} to \eqref{eq:71b} we have used that $\sin\theta$
vanished at $(0,\pi)$ and to pass from \eqref{eq:71c} to \eqref{eq:71d} we have
used that $q$ vanished at $(0,\pi)$. 
Collecting equations \eqref{eq:70} and \eqref{eq:71e}, from equation
\eqref{eq:69} we deduce our fundamental equation 
\begin{equation}
  \label{eq:23}
 \mf=8+\frac{2}{\pi} \int_{S^2} (\sigma+q) \, \dv.
\end{equation}
From equation \eqref{eq:23} we get an alternative expression for the first
variation of the mass 
\begin{equation}
  \label{eq:24}
   \mf'=\frac{2}{\pi} \int_{S^2} (\sigma'+ q') \, \dv. 
\end{equation}
And hence, using the first item in lemma \ref{l:1}, we get that 
\begin{equation}
  \label{eq:25}
   \int_{S^2} (\sigma'+ q')  \, \dv|_{\epsilon=0}  =0.
\end{equation}
Analogously, the second variation of the mass is given by
\begin{equation}
  \label{eq:72}
   \mf''=\frac{2}{\pi} \int_{S^2} (\sigma''+ q'') \, \dv.
\end{equation}
Using the second item in lemma \ref{l:1} we obtain
\begin{equation}
  \label{eq:73}
 \int_{S^2} (\sigma''+ q'') \, \dv|_{\epsilon=0}  >0.
\end{equation}
We are now ready to compute the first and second variation of the area. If we
evaluate the first variation of the area (equation \eqref{eq:27}) at
$\epsilon=0$ and use the (remarkable) fact that $e^{\sigma_0+q_0}$ is constant
for the extreme Kerr cylinder (see equation \eqref{eq:key}) we get
\begin{equation}
  \label{eq:29}
  A'|_{\epsilon=0} =4|J|  \int(\sigma'+ q') \, \dv|_{\epsilon=0}.
\end{equation}
Using \eqref{eq:25} we finally get
  \begin{equation}
    \label{eq:30}
     A'|_{\epsilon=0} =0. 
  \end{equation}
For the second variation we use equations \eqref{eq:28} and again equation
\eqref{eq:key} to obtain   
\begin{equation}
  \label{eq:28b}
  A''|_{\epsilon=0} = 4|J| \int_{S^2}((\sigma'+ q')^2 +( \sigma'' +  q'')) \,
  dS|_{\epsilon=0}. 
\end{equation}
The first term inside the integral is clearly positive definite and the second
also by \eqref{eq:72} and \eqref{eq:73}. Hence we deduce 
\begin{equation}
  \label{eq:31}
  A''|_{\epsilon=0} >0. 
\end{equation}
This concludes the proof of theorem \ref{t:main}.

\section{Application: spinning Bowen-York initial data}
\label{sec:aplic-spinn-bowen}
The Bowen-York initial data have been discovered in \cite{Bowen80} and since
that time they have been extensively used in both analytical and numerical
studies.  In this section we will prove that the area of the minimal surface
(which is also an apparent horizon) for the family of spinning Bowen-York
initial data satisfies the inequality \eqref{eq:5}. We will assume that the
inequality is true for extreme throat   initial data. As it was pointed out
in section \ref{sec:main-result} theorem \ref{t:main} suggests that this is the
case but the technical steps to complete the proof remain to be done.

 The argument runs as follows. In
\cite{Dain:2008yu} the extreme limit procedure was rigorously constructed for
this kind of data. The only property of this limit not proved in this article
was the monotonicity of the area. This is proved here as follows.

The area of any surface $r=constant$ is given by
\begin{equation}
  \label{eq:59}
  A_\mu(r)=2\pi r^2 \int_{S^2} \Phi_\mu^4\, \dv.
\end{equation}
We use the same notation $\Phi_\mu$ for the conformal factor of the Bowen-York
family used in \cite{Dain:2008yu}.

The location of the minimal surface (by the isometry of the data) is on
$r=\mu/2$. That is, we want to consider the area $A_\mu(\mu/2)$.  
By definition of minimal surface, we known that
\begin{equation}
  \label{eq:60}
  A_\mu(\mu/2)\leq A_\mu(r),
\end{equation}
for all $r$.  We also known that the conformal factor is monotonically
decreasing with $\mu$ (Lemma 3.2 in \cite{Dain:2008yu}). That is,  for
$\mu_1\leq \mu_2$ we have 
\begin{equation}
  \label{eq:61}
  \Phi_{\mu_1}(r,\theta)\leq \Phi_{\mu_2}(r,\theta).
\end{equation}
Hence we have
\begin{equation}
  \label{eq:63}
  A_{\mu_1}(r)\leq  A_{\mu_2}(r).
\end{equation}

Then we prove the following
\begin{equation}
  \label{eq:64}
  A_{\mu_1}(\mu_1/2)\leq  A_{\mu_1}(r) \leq  A_{\mu_2}(r),
\end{equation}
for all $r$. The first inequality in \eqref{eq:64} follows from \eqref{eq:60},
and the second from \eqref{eq:63}. That is, we have proved that any surface for
$\mu_2$ has bigger area than the minimal surface for $\mu_1$. In particular,
the minimal surface $r=\mu_2/2$, that is
\begin{equation}
  \label{eq:67}
   A_{\mu_1}(\mu_1/2)\leq A_{\mu_2}(\mu_2/2).
\end{equation}
 Note, however, that inequality \eqref{eq:64} is stronger than \eqref{eq:67}.

 We have proved that the area of the minimal surface is monotonically
 decreasing under the extreme limit process constructed in
 \cite{Dain:2008yu}. And hence the area of the related extreme throat initial  data
 (which can be also rigorously constructed, see \cite{gabach09}
 \cite{Hannam:2009ib}) is smaller than the original area of the minimal
 surface. Since the inequality holds on the extreme cylinder it follows that it
 also holds for the spinning Bowen-York initial data.

\section{Final comments}
\label{sec:final-comments}
The first open problem is to complete the analysis presented in the proof of
theorem \ref{t:main} and prove the inequality \eqref{eq:16} on extreme
cylindrical initial data. We expect that the proof will follow similar lines as
the ones presented in \cite{Dain05d} \cite{Dain06c} \cite{Costa:2009hn}
\cite{Chrusciel:2009ki} \cite{Chrusciel:2007ak}. We are currently working on
this \cite{dain10c}.

The second open problem, which is much more difficult and relevant, is to
construct an extreme limit procedure for generic axially symmetric initial data
which satisfies the properties i) and ii) mentioned in section
\ref{sec:main-result}. Then, the conjecture \ref{c:1} will be reduced to the
extreme throat initial data case and hence it will be proved.

\begin{acknowledgments} 
  It is a pleasure to thank Marc Mars for illuminating discussions regarding
  geometrical inequalities over many years. Particular useful for this work
  were the ones that took place at the Mathematisches Forschungsinstitut
  Oberwolfach during the workshop ``Mathematical Aspects of General
  Relativity'', October 11th -- October 17th, 2009 and during the conference
  ``PDEs, relativity \& nonlinear waves'', Granada, April 5-9, 2010. The author
  thanks the organizers of these events for the invitation and the hospitality
  and support of the Mathematisches Forschungsinstitut Oberwolfach.

  The author is supported by CONICET (Argentina).  This work was supported in
  part by grant PIP 6354/05 of CONICET (Argentina), grant 05/B415 Secyt-UNC
  (Argentina) and the Partner Group grant of the Max Planck Institute for
  Gravitational Physics, Albert-Einstein-Institute (Germany).

 \end{acknowledgments}

\appendix

\section{The extreme Kerr cylindrical initial data} 
\label{sec:extr-kerr-cylindr}
In this appendix we collect some well known properties of the extreme Kerr
black hole initial data defined by a slice $t=constant$ in the Boyer-Lindquist
coordinates. We use a slight modification of these coordinates. Let us denote by
$\tilde r$ the Boyer-Lindquist radius, we define $r=\tilde r -m$. In this way
the cylindrical end is located at $r=0$.

In the extreme case we have $\sqrt|J|=m$ (we always take the positive sign of
the square root). We denote by $a$ the standard angular momentum per unit mass
parameter $a=J/m$.  Note that for extreme Kerr we have two possible values for
the angular momentum $J=\pm m^2$, and hence $a=\pm m$. The spacetime metric
depends only one parameter, in our case it is appropriate to chose $J$ as the
free parameter.

The square of the norm $\eta$ of the axial Killing vector is given by 
\begin{equation}
  \label{eq:59bb}
 \eta =\frac{( (r+\sqrt{|J|})^2+|J|)^2-|J| r^2\sin^2\theta}{\Sigma}\sin^2\theta, 
\end{equation}
were  $\Sigma$ is given by
\begin{equation}
  \Sigma=(r+\sqrt{|J|})^2+|J|\cos^2\theta.
\end{equation}
The twist potential of the axial Killing vector is given by 
\begin{equation}\label{wcero}
  \omega =2J(\cos^3\theta-3\cos\theta)-\frac{2J
    |J|\cos\theta\sin^4\theta}{\Sigma}.
\end{equation}
The conformal factor $\Phi$ and the function $q$ that characterize the
intrinsic metric \eqref{thcero}  of the slice  are
given by  
\begin{equation}
\label{ficero}
 e^{2q}=\frac{\Sigma\sin^2\theta}{\eta}, \quad
 \Phi^4=\frac{\eta}{r^2\sin^2\theta}. 
\end{equation}

From these function we compute the following limits
\begin{align}
  \label{eq:56}
  \lim_{r\to 0} (\sqrt{r}\Phi) &=\varphi_0 =
  \left(\frac{4|J|}{1+\cos^2\theta}\right)^{1/4},\\ 
  \label{eq:57}
   \lim_{r\to 0} e^{q} &= e^{q_0} = \frac{1+\cos^2\theta}{2},\\
 \lim_{r\to 0}\omega & = \omega_0
 =-\frac{8J\cos\theta}{1+\cos^2\theta},  \label{eq:58} 
\end{align}
The function $\sigma_0$ is given by
\begin{equation}
  \label{eq:35b}
  \sigma_0=4 \ln \varphi_0=\ln(4|J|)-\ln(1+\cos^2\theta). 
\end{equation}
We have the relation
\begin{equation}
  \label{eq:6}
  \eta_0=\sin^2\theta e^\sigma=\sin^2\theta\varphi^4_0.
\end{equation}
From \eqref{eq:35b} and \eqref{eq:57} we deduce the following key equation
\begin{equation}
  \label{eq:key}
e^{\sigma_0+ q_0}=2|J|.  
\end{equation}

The mass functional $\mf$ defined by \eqref{eq:22} evaluated at the extreme
Kerr cylindrical initial data can also be calculated explicitly from these
expression. The calculus simplify by noting that using equation \eqref{eq:18} we
can directly compute the integral
\begin{equation}
  \label{eq:36}
  \int_0^\pi\left(\frac{|\partial_\theta
    \omega|^2}{\eta^2}\right) \sin\theta \, d\theta=4. 
\end{equation}
The integral for $\sigma_0$ yields
\begin{equation}
  \label{eq:37}
  \int_0^\pi \sigma_0   \sin\theta \, d\theta=2\ln(2)-2\ln(|J|)+4-\pi.
\end{equation}
Finally, for the other integral we get
\begin{equation}
  \label{eq:38}
   \int_0^\pi | \partial \sigma_0|^2    \sin\theta \, d\theta= -12+4\pi.
\end{equation}
Hence, we obtain
\begin{equation}
  \label{eq:39}
  \mf=8(\ln(2|J|)+1).
\end{equation}


\begin{thebibliography}{10}

\bibitem{Arnowitt62}
R.~Arnowitt, S.~Deser, and C.~W. Misner.
\newblock The dynamics of general relativity.
\newblock In L.~Witten, editor, {\em Gravitation: An Introduction to Current
  Research}, pages 227--265. Wiley, New York, 1962.

\bibitem{Ashtekar:1996cd}
Abhay Ashtekar, Jiri Bicak, and Bernd~G. Schmidt.
\newblock {Asymptotic structure of symmetry reduced general relativity}.
\newblock {\em Phys. Rev.}, D55:669--686, 1997.

\bibitem{lrr-2004-10}
Abhay~Ashtekar Badri~Krishnan.
\newblock Isolated and dynamical horizons and their applications.
\newblock {\em Living Reviews in Relativity}, 7(10), 2004.

\bibitem{Bardeen:1999px}
James~M. Bardeen and Gary~T. Horowitz.
\newblock {The extreme Kerr throat geometry: A vacuum analog of AdS(2) x S(2)}.
\newblock {\em Phys. Rev.}, D60:104030, 1999.

\bibitem{Bartnik86}
R.~Bartnik.
\newblock The mass of an asymptotically flat manifold.
\newblock {\em Comm. Pure App. Math.}, 39(5):661--693, 1986.

\bibitem{Bartnik04b}
Robert Bartnik and Jim Isenberg.
\newblock The constraint equations.
\newblock In Piotr~T. Chru\'sciel and Helmut Friedrich, editors, {\em The
  {E}instein equations and large scale behavior of gravitational fields}, pages
  1--38. Birh\"auser Verlag, Basel Boston Berlin, 2004.

\bibitem{beig97}
Robert Beig and Piotr~T. Chru{\'s}ciel.
\newblock Killing initial data.
\newblock {\em Classical Quantum Gravity}, 14(1A):A83--A92, 1997.
\newblock Geometry and physics.

\bibitem{Booth:2007wu}
Ivan Booth and Stephen Fairhurst.
\newblock {Extremality conditions for isolated and dynamical horizons}.
\newblock {\em Phys. Rev.}, D77:084005, 2008.

\bibitem{Bowen80}
Jeffrey~M. Bowen and James~W. York, Jr.
\newblock Time-asymmetric initial data for black holes and black-hole
  collisions.
\newblock {\em Phys. Rev. D}, 21(8):2047--2055, 1980.

\bibitem{Carter71}
B.~Carter.
\newblock Axisymmetric black hole has only two degrees of freedom.
\newblock {\em Phys. Rev. Lett.}, 26(6):331--333, 1971.

\bibitem{Christodoulou70}
D.~Christodoulou.
\newblock Reversible and irreversible transforations in black-hole physics.
\newblock {\em Phys. Rev. Lett.}, 25:1596--1597, 1970.

\bibitem{Chrusciel:2007dd}
Piotr~T. Chrusciel.
\newblock {Mass and angular-momentum inequalities for axi-symmetric initial
  data sets I. Positivity of mass}.
\newblock {\em Annals Phys.}, 323:2566--2590, 2008.

\bibitem{Chrusciel:2007ak}
Piotr~T. Chru{\'s}ciel, Yanyan Li, and Gilbert Weinstein.
\newblock Mass and angular-momentum inequalities for axi-symmetric initial data
  sets. {II}. {A}ngular-momentum.
\newblock {\em Ann. Phys.}, 323(10):2591--2613, 2008.

\bibitem{Chrusciel:2009ki}
Piotr~T. Chrusciel and Joao Lopes~Costa.
\newblock {Mass, angular-momentum, and charge inequalities for axisymmetric
  initial data}.
\newblock {\em Class. Quant. Grav.}, 26:235013, 2009.

\bibitem{Costa:2009hn}
João~Lopes Costa.
\newblock Proof of a {D}ain inequality with charge.
\newblock {\em Journal of Physics A: Mathematical and Theoretical},
  43(28):285202, 2010.

\bibitem{Dain99b}
Sergio Dain.
\newblock Initial data for a head on collision of two {K}err-like black holes
  with close limit.
\newblock {\em Phys. Rev. D}, 64(15):124002, 2001.

\bibitem{Dain05e}
Sergio Dain.
\newblock Angular momemtum-mass inequality for axisymmetric black holes.
\newblock {\em Phys. Rev. Lett.}, 96:101101, 2006.

\bibitem{Dain05d}
Sergio Dain.
\newblock Proof of the (local) angular momemtum-mass inequality for
  axisymmetric black holes.
\newblock {\em Class. Quantum. Grav.}, 23:6845--6855, 2006.

\bibitem{Dain05c}
Sergio Dain.
\newblock A variational principle for stationary, axisymmetric solutions of
  einstein's equations.
\newblock {\em Class. Quantum. Grav.}, 23:6857--6871, 2006.

\bibitem{Dain:2008xr}
Sergio Dain.
\newblock {Axisymmetric evolution of Einstein equations and mass conservation}.
\newblock {\em Class. Quantum. Grav.}, 25:145021, 2008.

\bibitem{Dain:2007pk}
Sergio Dain.
\newblock The inequality between mass and angular momentum for axially
  symmetric black holes.
\newblock {\em International Journal of Modern Physics D}, 17(3-4):519--523,
  2008.

\bibitem{Dain06c}
Sergio Dain.
\newblock Proof of the angular momentum-mass inequality for axisymmetric black
  holes.
\newblock {\em J. Differential Geometry}, 79(1):33--67, 2008.

\bibitem{Dain:2010uh}
Sergio Dain and Maria E.~Gabach Clement.
\newblock {Small deformations of extreme Kerr black hole initial data}, 2010.

\bibitem{dain10c}
Sergio Dain and María Eugenia~Gabach Cl\'ement, 2010.
\newblock In preparation.

\bibitem{Dain:2008yu}
Sergio Dain and Mar\'ia~Eugenia Gabach~Cl\'ement.
\newblock {Extreme Bowen-York initial data}.
\newblock {\em Class. Quantum. Grav.}, 26:035020, 2009.

\bibitem{dain10}
Sergio Dain and Omar~E. Ortiz.
\newblock Well-posedness, linear perturbations, and mass conservation for the
  axisymmetric einstein equations.
\newblock {\em Phys. Rev. D}, 81(4):044040, Feb 2010.

\bibitem{gabach09}
Mar\'ia~Eugenia Gabach~Cl\'ement.
\newblock {Conformally flat black hole initial data, with one cylindrical end},
  2009.

\bibitem{Hannam:2009ib}
Mark Hannam, Sascha Husa, and Niall~O Murchadha.
\newblock {Bowen-York trumpet data and black-hole simulations}.
\newblock {\em Phys. Rev.}, D80:124007, 2009.

\bibitem{hennig08}
J\"org Hennig, Marcus Ansorg, and Carla Cederbaum.
\newblock A universal inequality between the angular momentum and horizon area
  for axisymmetric and stationary black holes with surrounding matter.
\newblock {\em Class. Quantum. Grav.}, 25(16):162002, 2008.

\bibitem{Hennig:2008zy}
J\"org Hennig, Carla Cederbaum, and Marcus Ansorg.
\newblock {A universal inequality for axisymmetric and stationary black holes
  with surrounding matter in the Einstein-Maxwell theory}.
\newblock {\em Commun. Math. Phys.}, 293:449--467, 2010.

\bibitem{Jang79}
Pong~Soo Jang.
\newblock Note on cosmic censorship.
\newblock {\em Phys. Rev. D}, 20(4):834--837, 1979.

\bibitem{Mars:2009cj}
Marc Mars.
\newblock {Present status of the Penrose inequality}.
\newblock {\em Class. Quant. Grav.}, 26:193001, 2009.

\bibitem{Moncrief75}
Vincent Moncrief.
\newblock Spacetime symmetries and linearization stability of the {E}instein
  equations. {I}.
\newblock {\em J. Math. Phys.}, 16:493--498, 1975.

\bibitem{Schoen79b}
Richard Schoen and Shing~Tung Yau.
\newblock On the proof of the positive mass conjecture in general relativity.
\newblock {\em Comm. Math. Phys.}, 65(1):45--76, 1979.

\bibitem{Schoen81}
Richard Schoen and Shing~Tung Yau.
\newblock Proof of the positive mass theorem. {II}.
\newblock {\em Comm. Math. Phys.}, 79(2):231--260, 1981.

\bibitem{stephani03}
Hans Stephani, Dietrich Kramer, Malcolm MacCallum, Cornelius Hoenselaers, and
  Eduard Herlt.
\newblock {\em Exact solutions of {E}instein's field equations}.
\newblock Cambridge Monographs on Mathematical Physics. Cambridge University
  Press, Cambridge, second edition, 2003.

\bibitem{Szabados04}
L{\'a}szl{\'o}~B. Szabados.
\newblock Quasi-local energy-momentum and angular momentum in {GR}: A review
  article.
\newblock {\em Living Rev. Relativity}, 7(4), 2004.
\newblock cited on 8 August 2005.

\bibitem{Weinstein05}
Gilbert Weinstein and Sumio Yamada.
\newblock On a {P}enrose inequality with charge.
\newblock {\em Commun. Math. Phys.}, 257(3):703--723, 2005.

\bibitem{Witten81}
Edward Witten.
\newblock A new proof of the positive energy theorem.
\newblock {\em Commun. Math. Phys.}, 80(3):381--402, 1981.

\end{thebibliography}

\end{document}